\documentclass[a4paper,journal]{IEEEtran}  
\usepackage{amsmath}
\usepackage{color}
\usepackage{float}
\usepackage{amsmath}
\usepackage{amssymb}
\usepackage{mathrsfs}
\usepackage{mathtools}
\usepackage[table]{xcolor}
\usepackage{arydshln}
\usepackage{amsthm}
\usepackage{booktabs}
\usepackage{cite}
\usepackage{pifont}
\usepackage{amssymb}
\usepackage{multirow}
\usepackage{mathtools}
\usepackage{breqn}

\usepackage{graphicx}

\makeatletter

\usepackage{algorithm}
\usepackage{algpseudocode}

\usepackage{algorithmicx}
\usepackage{algpseudocode}
\usepackage{amsmath,amsfonts,amssymb}
\usepackage{amssymb}

\usepackage{array,booktabs,arydshln}
\usepackage{arydshln}

\usepackage{float} 

\usepackage[caption=false,font=footnotesize]{subfig}

\@ifundefined{showcaptionsetup}{}{%
 \PassOptionsToPackage{caption=false}{subfig}}

\usepackage{subfig}
\usepackage{adjustbox}
\usepackage{multicol}
\usepackage{multirow}

\makeatother

\newcommand{\hidden}[1]{}

\usepackage{tikz}
\usepackage{tablefootnote}
\usepackage{booktabs} 
\usepackage{graphicx} 
\usepackage{pifont} 
\usepackage{array, colortbl,multirow}

\usetikzlibrary{matrix,shapes.arrows}
\renewcommand{\arraystretch}{1.2}

\newtheorem{theorem}{\bf {Theorem}}

\usepackage{tikz}

\usetikzlibrary{matrix,shapes.arrows}
\renewcommand{\arraystretch}{1.2}

\usepackage[justification=centering]{caption}
\usepackage{hyperref}
\hypersetup{
	unicode=false,          
	pdftoolbar=true,        
	pdfmenubar=true,        
	pdffitwindow=false,     
	pdfstartview={FitH},    
	pdftitle={My title},    
	pdfauthor={Author},     
	pdfsubject={Subject},   
	pdfcreator={Creator},   
	pdfproducer={Producer}, 
	pdfkeywords={keyword1, key2, key3}, 
	pdfnewwindow=true,      
	colorlinks=true,       
	linkcolor=red,          
	citecolor=red,        
	filecolor=red,      
	urlcolor=red           
}
\newfloat{function}{htbp}{lof}
\floatname{function}{Function}
\begin{document}
	\markboth{IEEE Submission}%
	{Nguyen \MakeLowercase{\textit{et al.}}: }
	\title{Exact DC Representation of Multi-Tier Offloading Product in SAGINs via Quantifier Elimination}
	\author{
		Minh-Tuong Nguyen, and Vo Phi Son
		\thanks{M.-T. Nguyen and V. P. Son are with the Smart Green Transformation Center (GREEN-X) and the College of Engineering and Computer Science, VinUniversity, Vinhomes Ocean Park, Hanoi 100000, Vietnam (emails: \{tuong.nm, son.vp\}@vinuni.edu.vn).He is also with the Faculty of Engineering and Information Technology, University of Technology Sydney, Sydney, NSW 2007, Australia (e-mail: tuong.nguyen@student.uts.edu.au).}%
	}%
	\maketitle
	\begin{abstract}
		Task offloading in space--air--ground integrated networks (SAGIN) yields non-convex signomial or polynomial programs with cubic couplings. Sequential geometric programming (SGP) approximates them via exponential cone representations, which exceeds the second-order cone programming (SOCP) ceiling of embedded code generators such as CVXPYgen. We derive a difference-of-convex (DC) representation exactly certified over the reals by quantifier elimination and apply the convex--concave procedure (CCP), whose SOCP subproblems remove this structural obstacle to future embedded code generation. Comparisons with the BARON global solver show that SGP and CCP both attain near-global solutions. CCP further reduces the average solution time from SGP's $0.1012$~s to $0.0113$~s, an $8.9$-fold speedup.
	\end{abstract}
	\IEEEpeerreviewmaketitle{}
	\section{Introduction} 
	\label{sec:Introduction}
	Satellite communication enables global connectivity by extending coverage to remote regions~\cite{le2025performance,jung2023satellite,maral2020satellite}. Meanwhile, multi-access edge computing (MEC) moves computation closer to data sources to meet stringent delay and bandwidth requirements~\cite{shi2022delay,huang2024joint}.
In space--air--ground integrated networks (SAGINs), MEC extends across ground, aerial, and satellite nodes. These architectures incorporate unmanned aerial vehicle (UAV) relays, satellite backhaul, and inter-satellite links (ISLs) to improve end-to-end offloading performance by extending coverage even further~\cite {huang2024joint}.
   However, multi-tier offloading couples the partial offload decision assigned to different tiers. This coupling yields non-convex signomial terms, specifically cubic products of decision variables~\cite{huang2024joint}.
    These products are signomials, making it practical to approximate this signomial problem via sequential geometric programming (SGP)~\cite{van2018joint,filabadi2024exponential}.
    In SGP, each subproblem is a GP~\cite{van2018joint}, whose standard conic form requires exponential-cone transformations~\cite{filabadi2024exponential}. This creates a deployment gap for embedded optimization: optimization models with logarithmic or exponential functions and exponential-cone constraints are slower or less reliable to solve~\cite{chen2023mixed}, whereas CVXPYgen mainly targets embedded code generation for second-order cone programming (SOCP) models~\cite{schaller2022embedded}.
	Accordingly, deployment-oriented SAGIN offloading motivates an approximation whose subproblems are SOCPs.
	Therefore, we formulate a quantifier-elimination (QE) problem~\cite{davenport1988real} that yields a difference-of-convex (DC) representation for the coupled offloading signomials with an exact certificate over the reals.
	We then apply the convex--concave procedure (CCP)~\cite{lipp2016variations} to the transformed SAGIN offloading problem. Because the DC components are SOC-representable, each CCP iteration reduces to an SOCP.
	As such, this work removes the exponential-cone requirement at the optimization model level and provides a basis for future embedded implementations. Numerical results are provided to demonstrate its computational efficiency in terms of solution time.

	\section{System Model and Problem Formulation}
	\subsection{System Model}
	\begin{figure}
		\centering
		\includegraphics[width=\columnwidth]{"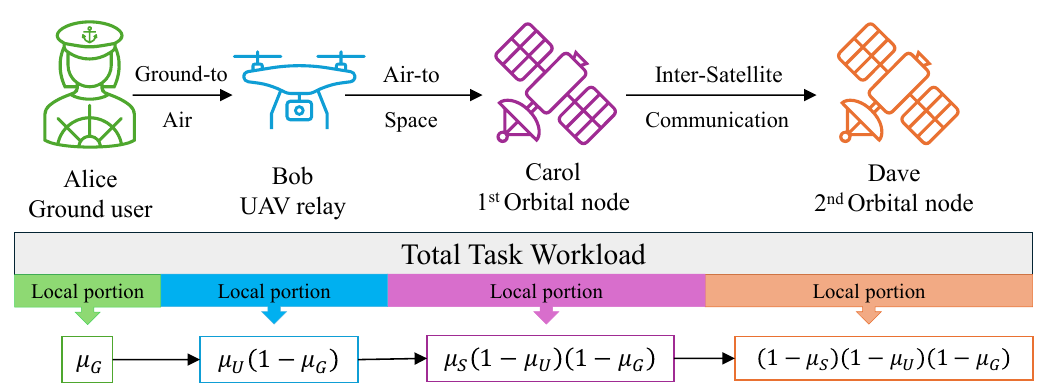"}
		\caption{Representative multi-tier SAGIN offloading chain with partial offloading decisions ${\mu _G}$, ${\mu_U}$, and ${\mu_S}$.}
		\label{fig:system-model}
	\end{figure}
	\begin{figure}
		\centering
		\includegraphics[width=0.85\columnwidth,trim={0 20pt 0 5pt},clip]{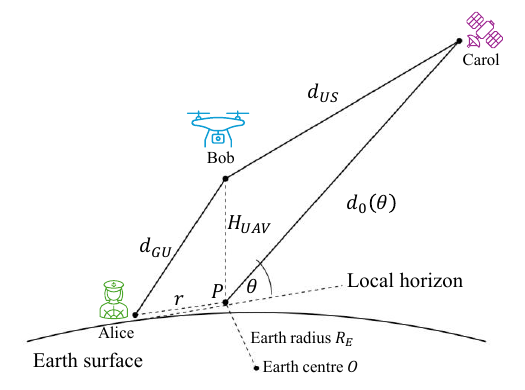}
		\caption{System geometry.}
		\label{fig:geometrybai2dc}
	\end{figure}
We consider a maritime vehicular SAGIN in which terrestrial edge infrastructure is unavailable, and sparse node availability yields a single feasible relay path over the considered task interval. As illustrated in Fig.~\ref{fig:system-model}, we study one active end-to-end offloading chain comprising a ground user (Alice, $G$), a UAV relay and compute node (Bob, $U$), and two LEO compute nodes (Carol, $S_1$; Dave, $S_2$). The task originated from Alice and is partially offloaded along the chain $G \to U \to S_1 \to S_2$. This chain illustrates the coupling among the offloading decisions across multiple SAGIN tiers.
	As shown in Fig.~\ref{fig:geometrybai2dc}, the Alice to Bob (GU) distance is $d_{GU}=\sqrt{H_{\mathrm{UAV}}^{2}+r^{2}}$, where $H_{\mathrm{UAV}}$ is the UAV altitude and $r$ is the horizontal offset between $G$ and $U$. Carol orbits at an altitude $h_{\mathrm{LEO}}$ above a spherical Earth of radius $R_E$, and is observed at an elevation angle $\theta$ measured from the local horizon at the sub-UAV point $P$ on the Earth's surface. From the geometry of the model, the surface-to-satellite distance is approximated as $d_{0}(\theta)\approx\sqrt{(R_E\sin\theta)^{2}+2R_Eh_{\mathrm{LEO}} +h_{\mathrm{LEO}}^{2}}-R_E\sin\theta$. Then, the UAV-to-satellite distance is $d_{\mathrm{US}} \approx \sqrt{d_0^2 + H_{\mathrm{UAV}}^2 - 2d_0H_{\mathrm{UAV}}\sin\theta}$.
	The ISL distance between $S_1$ and
	$S_2$ is $d_{\mathrm{ISL}}$.
	The Alice to Bob channel is modeled as a Rician fading channel, with $h_{GU}=\sqrt{\xi/(\xi+1)}e^{j\phi_{GU}}d_{GU}^{-\alpha_L/2}+\sqrt{1/(\xi+1)}\hat g d_{GU}^{-\alpha_N/2}$, where $\xi$ is the Rician factor, $\alpha_L$ and $\alpha_N$ are the line-of-sight (LoS) and non-line-of-sight (NLoS) exponents, $\hat g\sim\mathcal{CN}(0,1)$, and $\phi_{GU}\sim\mathcal{U}[0,2\pi)$~\cite{huang2024joint,mckay2005capacity}. Its rate is $R_{GU}=B_{GU}\log_2(1+P_G|h_{GU}|^2/\sigma_{GU}^2)$.
	The Bob to Carol channel coefficient is $h_{US}=\sqrt{\delta\lambda/(4\pi d_{US})}e^{j\phi_{US}}$, where $\delta$ is the beam gain and $\phi_{US}\sim\mathcal{U}[0,2\pi)$. Thus, $R_{US}=B_{US}\log_2(1+P_U|h_{US}|^2/\sigma_{US}^2)$ is deterministic for fixed geometry since $|h_{US}|^2=\delta\lambda/(4\pi d_{US})$.
	For the Carol to Dave ISL, $\mathrm{SNR}_{\mathrm{ISL}}=P_SG_tG_rL_p/[k_BT_{\mathrm{sys}}B_{\mathrm{ISL}}(4\pi d_{\mathrm{ISL}}f_c/c)^2]$, where $G_t$ and $G_r$ are antenna gains and $L_p$ is the pointing loss~\cite[\S~5.12]{maral2020satellite}. The ISL rate is $R_{\mathrm{ISL}}=B_{\mathrm{ISL}}\log_2(1+\mathrm{SNR}_{\mathrm{ISL}})$.
	From Alice, a task of $W$ bits and required task complexity $F$ in cycles is processed by the chain. Let $\mu_G,\mu_U,\mu_S\in[0,1]$ be the offload fractions processed at Alice, Bob, and Carol. Bob computes these fractions onboard. He terminates the link from Alice and
	originates the link to Carol, so $R_{GU}$, $R_{US}$, and $d_{US}$ are available
	at Bob. We assume that
	Bob returns $\mu_G$ to Alice before she transmits, applies $\mu_U$ locally, and
	appends $\mu_S$ to the payload forwarded to Carol, so the offloading decisions travel with the task itself rather than in separate control messages.
	Having this, the computation delays are adopted from~\cite{huang2024joint} as follows: $T_G^C=\mu_GF/f_G$, $T_U^C=\mu_U(1-\mu_G)F/f_U$, $T_{S_1}^C=\mu_S(1-\mu_U)(1-\mu_G)F/f_{S_1}$, and $T_{S_2}^C=(1-\mu_S)(1-\mu_U)(1-\mu_G)F/f_{S_2}$.
	With transmission overhead $b$, the transmission delays are $T_{GU}^T=b(1-\mu_G)W/R_{GU}$, $T_{US}^T=d_{US}/c+b(1-\mu_U)(1-\mu_G)W/R_{US}$, and $T_{S_1S_2}^T=d_{\mathrm{ISL}}/c+b(1-\mu_S)(1-\mu_U)(1-\mu_G)W/R_{\mathrm{ISL}}$~\cite{huang2024joint}.
	\subsection{Problem Formulation}
	We optimize the worst-case delay $T=\max(T_A,T_B,T_C,T_D)$, where $T_A=T_G^C$, $T_B=T_{GU}^T+T_U^C$, $T_C=T_{GU}^T+T_{US}^T+T_{S_1}^C$, and $T_D=T_{GU}^T+T_{US}^T+T_{S_1S_2}^T+T_{S_2}^C$~\cite{huang2024joint}. These four delay terms correspond to the possible completion points of the task, namely Alice, Bob, Carol, and Dave. Introducing the epigraph variable $z$, the worst-case delay minimization problem is written as
	\begingroup
	\allowdisplaybreaks
	\begin{subequations}
		\label{eq:optimization_problem_example}
		\begin{IEEEeqnarray}{rCl}
			&& \min_{\substack{\mu_{U},\, \mu_{G},\, 
					\mu_{S},\, z}}
			\quad z
			\label{eq:objective_example} \\
			&& \mathrm{s.t.}\quad
			T_{\mathcal{X}} - z \le 0, \quad \forall \mathcal{X} \in \{A,B,C,D\}
			\label{cons:Tg1} \\
			&& \phantom{\mathrm{s.t.}}\quad
			\mu_{U},\, \mu_{G},\, \mu_{S} \in [0,1].
			\label{cons:bounds}
		\end{IEEEeqnarray}
	\end{subequations}
	\endgroup
	The non-convexity of problem~\eqref{eq:optimization_problem_example} is induced by coupling decisions along the selected Alice-to-Dave route. This dependency introduces coupling among offloading decisions ${\mu _G}$, ${\mu_U}$, and ${\mu_S}$, resulting in signomial terms that are not directly compatible with convex optimization methods. The considered chain produces the cubic instance for our numerical evaluation, while Theorem~\ref{theorem_1} provides an exact representation for a general $N$-factor sequential coupling.
	\section{Proposed Approach}
	Motivated by the broad representational capacity of DC formulations for non-convex problems~\cite{lipp2016variations}, we address \eqref{eq:optimization_problem_example} through DC programming. This approach requires expressing all non-convex terms in \eqref{cons:Tg1} as $f(x)=g(x)-h(x)$~\cite{lipp2016variations}. Such DC representations are often non-trivial. For example, $f(x)=3x^3{+}3x$ admits the DC representation ${(\sqrt[4]{3/8}x{+}\sqrt[4]{3/8})^4} - {(\sqrt[4]{3/8}x{-}\sqrt[4]{3/8})^4}$, although $f(x)$ contains no irrational terms. Finite-precision interior-point methods cannot provide exact certificates for these DC representations~\cite{ahmadi2018dc,bomze2004undominated,murti2024lp}. A QE method such as Cylindrical Algebraic Decomposition (CAD) provides exact certification through minimal polynomials and root indices~\cite{maaz2025new,davenport1988real}. Accordingly, Section~\ref{sec:Quantifier Elimination Formulation} formulates the exact DC representation of the coupled multi-tier offloading products in problem~\eqref{eq:optimization_problem_example} as a QE problem solvable by CAD~\cite{maaz2025new}. Section~\ref{sec:ortho} then introduces a numerically stable change of variables, yielding a formulation solvable by CCP in Section~\ref{sec:final_solve}.
	\subsection{QE Formulation for Exact DC Representation of Multi-Tier Offloading Product} \label{sec:Quantifier Elimination Formulation}
	Based on the simple observation that $x_1x_2 = \frac{1}{4}(x_1+x_2)^2 - \frac{1}{4}(x_1-x_2)^2$ for $N=2$, we construct an ansatz for the general case by combining all exponentiation instances of $u \triangleq x_1 \pm x_2 \pm \cdots \pm x_N$ with undetermined coefficients $d_i$.  Therefore, we formulate the following QE problem $\forall \{x_1,x_2\} \,\left[ {{x_1}{x_2} = {d_1}{{\left( {{x_1} + {x_2}} \right)}^2} + {d_2}{{\left( {{x_1} - {x_2}} \right)}^2}} \right]$ for $N=2$. Solving it using CAD yields ${d_1} = \frac{1}{4}$ and ${d_2}= -\frac{1}{4}$, which confirms our ansatz. For $N=4$, we obtain $x_1x_2x_3x_4 = (1/192) \sum_{s_2, s_3, s_4 \in \{\pm 1\}} s_2 s_3 s_4 (x_1 + s_2 x_2 + s_3 x_3 + s_4 x_4)^4$. Note that a DC representation is only possible for even $N$. In particular, all $u^N$ terms can be expressed as $\left\lceil {{{\log }_2}N} \right\rceil $ cone constraints via the following recursive SOCP reformulation:  
	\begin{equation}
		\label{socp}
		u^N \le t
		\;\Longleftrightarrow\;
		\left\{
		\begin{aligned}
			y^2 &\le t\,u^{N\bmod 2},\\
			u^{\lceil N/2\rceil} &\le y.
		\end{aligned}
		\right.
	\end{equation}
	For the multi-tier offload coupling product, we first introduce the variables ${a_n} \! \in \! \left\{ {0,1} \right\}$ to determine the sign of each $x_n$ within $u$. Let  $\Xi \!=\! \sum_{n=2}^N (-1)^{a_n} x_n$ and $i =  1 + \sum_{n = 2}^N {{a_n}} {2^{N - n}}$. The QE formulation for exact DC representation is
	\begin{align}
		\forall \left\{ {{x_n}} \right\}_{n = 1}^N\prod\limits_{n = 1}^N {{x_n}}  = \sum_{{a_2} = 0}^1 {\sum_{{a_3} = 0}^1  \cdots  } \sum_{{a_N} = 0}^1 {{d_i}} {\left( {{x_1} + \Xi } \right)^N} \label{mainQE}
	\end{align}
	Unlike finite-precision constructions~\cite{ahmadi2018dc,bomze2004undominated,murti2024lp}, our QE approach provides an exact certificate, albeit with a doubly exponential complexity caveat~\cite{davenport1988real}. Therefore, the algebraic pattern of the verified solutions for $N=2$, and $N=4$ is generalized in Theorem~\ref{theorem_1} to avoid the aforementioned complexity.
	\begin{theorem} \label{theorem_1}
		Let $\Omega \!=\! (-1)^{\sum_{n=2}^N a_n}$ and $S_N \!=\! \sum_{a_2=0}^1 \sum_{a_3=0}^1 \cdots \sum_{a_N=0}^1 \Omega\,(x_1 {+} \Xi)^N$. The solution of \eqref{mainQE} yields the exact DC representation of the multi-tier offload-coupling product (first equality in \eqref{exact_solution}). Its signomial counterpart follows by sign inversion (second equality in \eqref{exact_solution}).
		\begin{alignat}{2}
			\prod\limits_{n = 1}^N {x_n} &= \frac{S_N}{2^{N - 1}N!}; \label{exact_solution}
			\qquad&
			-\prod\limits_{n = 1}^N {x_n} &= \frac{-S_N}{2^{N - 1}N!}. 
		\end{alignat}
	\end{theorem}
	\begin{proof}
		See Appendix~\ref{Exact_proof}.
	\end{proof}	
	For completeness, we also propose a binary presentation for the first equality in Eq.~\eqref{exact_solution} as shown in Table~\ref{tab:parity_functions}. To collect terms for $g(x)$, we need $\sum_{n=2}^{N} a_n$ to be even, which occurs when an even number of the $a_n$ values equal $1$. Conversely, for $h(x)$, this sum must be odd. This can be done by first listing all distinct binary assignments for $a_2, a_3, \dots, a_{N-1}$ following the natural decimal order ($0,1,2,3,\ldots$). The final bit $a_N$ can be determined to satisfy the parity constraint. Specifically, if an even parity is required, we set $a_N = 0$ when the partial sum $S'  =  a_2  +  a_3  +  \dots  +  a_{N-1}$ is even, or $a_N = 1$ when $S'$ is odd. For odd parity, the assignment of $a_N$ is reversed through NOT logic by the $a_4$ column for the lower half of Table~\ref{tab:parity_functions}.
	\begin{table}
		\centering
		\caption{\textnormal{Bit values $a_j$ $(j \ge 2)$ for $h(x)$ and $g(x)$, $N = 4$}}
		\label{tab:parity_functions}
		\renewcommand{\arraystretch}{1.2}
		\begin{tabular}{|c| c c |c| c|}
			\hline
			\rowcolor[gray]{0.93}
			& \multicolumn{2}{c|}{$a_2, a_3$} & \textbf{$a_4$} & Parity \\ 
			\hline
			0 & 0 & 0 & \multicolumn{1}{c|}{0} & \multirow{4}{*}{\begin{tabular}[c]{@{}c@{}}Even parity\\ for $g(x)$\end{tabular}} \\ 
			\cline{4-4}
			1 & 0 & 1 & \multicolumn{1}{c|}{1} & \\ 
			\cline{4-4}
			2 & 1 & 0 & \multicolumn{1}{c|}{1} & \\ 
			\cline{4-4}
			3 & 1 & 1 & \multicolumn{1}{c|}{0} & \\ 
			\hline
			\hline
			0 & 0 & 0 & \multicolumn{1}{c|}{1} & \multirow{4}{*}{\begin{tabular}[c]{@{}c@{}}Odd parity\\ for $h(x)$\end{tabular}} \\ 
			\cline{4-4}
			1 & 0 & 1 & \multicolumn{1}{c|}{0} & \\ 
			\cline{4-4}
			2 & 1 & 0 & \multicolumn{1}{c|}{0} & \\ 
			\cline{4-4}
			3 & 1 & 1 & \multicolumn{1}{c|}{1} & \\ 
			\cline{4-4}
			\hline
		\end{tabular}
	\end{table}
	\subsection{Hadamard Change of Variables} \label{sec:ortho}
	First, we introduce the column vector $\mathbf{t} = \left[ {{t_1},{t_2}, \cdots ,{t_{{2^{N {-} 1}}}}} \right]^\top$, where each entry of $\mathbf{t}$ stores a linear combination of the form 
	$x_1 \pm x_2 \pm \cdots\pm x_N$. Thus, each entry of $\mathbf{t}$ corresponds to a term of
	$x_1 + \Xi$ within the sum $S_N$. For $N=4$, $\mathbf{t}$ can be written as $\mathbf{t} = \mathbf{A_4} \mathbf{x}$ where $\mathbf{x} = [x_1, x_2, x_3, x_4]^\top$. 
	The matrix $\mathbf{A_4}$ is partitioned as $\mathbf{A_4} = [\mathbf{A_4^U}, \mathbf{A_4^L}]^\top$, 
	where both $\mathbf{A_4^U}$ and $\mathbf{A_4^L}$ have a Hadamard structure~\cite{horadam2012hadamard}.
	\begin{equation}
		\resizebox{0.90\columnwidth}{!}{$
			\mathbf{A_4^U}= \left[ {\begin{array}{*{20}{c}}
					1  &  1  &  1  &  1  \\
					1  &  1  & -1  & -1  \\
					1  & -1  &  1  & -1  \\
					1  & -1  & -1  &  1  \\
			\end{array}} \right]; \quad
			\mathbf{A_4^L} = \left[ {\begin{array}{*{20}{c}}
					1  &  1  &  1  & -1  \\
					1  &  1  & -1  &  1  \\
					1  & -1  &  1  &  1  \\
					1  & -1  & -1  & -1  
			\end{array}} \right]
			$}.
	\end{equation}
	Specifically, the upper half $\mathbf{A_4^U}$ follows the Sylvester construction of $\mathbf{H_4} = \mathbf{H_2} \otimes \mathbf{H_2}$ with additional row permutations that lead to the efficient back substitution $\mathbf{x} = \frac{1}{4}\mathbf{A_4^U}\left[ {{t_1},{t_2},{t_3},{t_4}} \right]^\top$~\cite[\S$2$]{horadam2012hadamard}. Its Hadamard equivalent lower half $\mathbf{A_4^L}$ is obtained by multiplying the last column of $\mathbf{A_4^U}$ by ${-}1$ (bit inversion of column $a_4$ in Table~\ref{tab:parity_functions}), and hence $\mathbf{x} = \frac{1}{4}{\left( \mathbf{A_4^L} \right)^\top}{\left[ {{t_5},{t_6},{t_7},{t_8}} \right]^\top}$ \cite[\S$2$]{horadam2012hadamard}. As such, we have
	\begin{align} 
		\frac{1}{4}\mathbf{A_4^U}\left[ {{t_1},{t_2},{t_3},{t_4}} \right]^\top = \frac{1}{4}(\mathbf{A_4^L})^\top\left[ {{t_5},{t_6},{t_7},{t_8}} \right]^\top \label{eq: doi_bien}
	\end{align}
	For the general case, we consider the matrix $\mathbf{A_N} \in \{ \pm 1 \}^{2^{N-1} \times N}$. This means that $\mathbf{A_N}$ has $2^{N-1}$ rows and $N$ columns, and every entry is $\pm 1$. The first column is a vector of ones with length $2^{N-1}$. For each column $j$, with $j = 2,\dots,N$, the entries alternate between  $\pm 1$. These alternating signs are written as $(-1)^{a_j}$, where $a_j$ is defined in the header of Table~\ref{tab:parity_functions}. 
	\textcolor{blue}{}
	\begin{theorem} \label{theorem_2}
		The change of variables $\mathbf{t} = \mathbf{A_N} \mathbf{x}$, characterized by the full-rank column-orthogonal matrix $\mathbf{A_N}$, is optimally conditioned with condition number $\kappa(\mathbf{A_N})=1$.
	\end{theorem}
	\begin{proof}
		See Appendix~\ref{matrix_proof}.
	\end{proof}	
	\subsection{Preprocessing for Convex--Concave Procedure} \label{sec:final_solve}
	Equation~\eqref{exact_solution} is used to obtain the DC representation of the coupling decision variables in \eqref{cons:Tg1}. 
	Next, the quartic terms from the aforementioned DC representation are reformulated via~\eqref{socp}. This transforms problem~\eqref{eq:optimization_problem_example} into a QCQP.
    Finally, all concave quadratic (or quartic) parts are linearized, and Problem~\eqref{eq:optimization_problem_example} is approximated using CCP (Algorithm~1.1 in~\cite{lipp2016variations}), as recommended in~\cite{park2017general}.
    Here, each iteration of CCP involves only an SOCP thanks to~\eqref{socp}.

	To begin with, we introduce the partition $\mathbf{t} \!=\! \left[ {\mathbf{v_1}, \mathbf{v_2}} \right]^\top$, where $\mathbf{v_1} \!=\! {\left[ {{t_1},{t_2},{t_3},{t_4}} \right]^ \top }$, $\mathbf{v_2} \!=\! {\left[ {{t_5},{t_6},{t_7},{t_8}} \right]^ \top }$, $\mathbf{x} \!=\! \left[ {1 - {\mu_{S}},1 - {\mu _{{U}}},1 - {\mu _{{G}}},1} \right]^\top$, $\mathbf{v_1^L} = {\left[ {0,0,0,1} \right]^ \top }$ and $\mathbf{v_1^U} = {\left[ {1,1,1,1} \right]^ \top}$. Using~\eqref{eq: doi_bien} and Theorem~\ref{theorem_2}, we use the change of variable $\mathbf{x} \!=\! \frac{1}{4}\mathbf{A_4^U}\mathbf{v _1} \!= \! \frac{1}{4}{\left( \mathbf{A_4^L} \right)^ \top }\mathbf{v_2}$ to express \eqref{cons:Tg1} in DC form. As such, $(1 {-} {\mu_{S}})(1 {-} {\mu _{{U}}})(1 {-} {\mu _{{G}}})1$ can now be written as ${{{192}^{ {-} 1}}\sum\limits_{i = 1}^4 {t_i^4}  {-} {{192}^{ {-} 1}}\sum\limits_{i = 5}^8 {t_i^4} }$. Then, the remaining terms $1 {-} {\mu_{S}}$, $1 {-} {\mu _{{U}}}$, $1 {-} {\mu _{{G}}}$, $\left( {1 {-} {\mu _{{U}}}} \right)\left( {1 {-} {\mu _{{G}}}} \right)$, ${\mu _{{U}}}\left( {1 {-} {\mu _{{G}}}} \right)$, and ${\mu_{S}}(1 {-} {\mu _{{U}}})(1 {-} {\mu _{{G}}})$ can now be written in linear or DC forms with respect to the entries of $\mathbf{t}$ as ${4^{-1}}\left( {{t_1} {+} {t_2} {+} {t_3} {+} {t_4}} \right)$, ${4^{-1}}\left( {{t_1} {+} {t_2} {-} {t_3} {-} {t_4}} \right)$, ${4^{-1}}\left( {{t_1} {-} {t_2} {+} {t_3} {-} {t_4}} \right)$, ${16^{-1}}\left[ {{{\left( {{t_1} {-} {t_4}} \right)}^2} - {{\left( {{t_2} {-} {t_3}} \right)}^2}} \right]$, ${4^{ - 1}}\left( {{t_1} {-} {t_2} {+} {t_3} {-} {t_4}} \right) {-} {16^{ - 1}}\left[ {{{\left( {{t_1} {-} {t_4}} \right)}^2} {-} {{\left( {{t_2} {-} {t_3}} \right)}^2}} \right]$, and ${16^{-1}}{\left( {{t_1} {-} {t_4}} \right)^2} {+} {192^{-1}}\sum\limits_{i = 5}^8 {t_i^4}  {-} {16^{-1}}{\left( {{t_2} {-} {t_3}} \right)^2} {-} {192^{-1}}\sum\limits_{i = 1}^4 {t_i^4}$, respectively. By substituting all these terms into \eqref{cons:Tg1}, we can write \eqref{cons:Tg1} in DC form as \eqref{cons:DC_form}. In addition, all aforementioned quartic $t_i^4$ terms can be reformulated as SOCP constraints using~\eqref{socp}. 
	Let $\preceq$ denote the component-wise less-than-or-equal inequality which follows the convention in \cite[\S$2.2.4$]{boyd2004convex}. After that \eqref{cons:bounds} is expressed as the linear constraint $\mathbf{v_1^L} \!\preceq \! \mathbf{A_4^U} \,\mathbf{v_1}\preceq \mathbf{v_1^U}$. Then, problem~\eqref{eq:optimization_problem_example} is transformed into a DC QCQP problem with SOCP subproblems at each CCP iteration.
	\begingroup
	\allowdisplaybreaks
	\begin{subequations}
		\label{eq:optimization_problem_transformed}
		\begin{IEEEeqnarray}{rCl}
			&& \min_{\mathbf{t},\, z} \quad z
			\label{eq:objective_transformed} \\
			&& \mathrm{s.t.}\quad
			g_{\mathcal{X}}\left( \mathbf{t},z \right) - h_{\mathcal{X}}\left( \mathbf{t},z \right) \le 0,  \forall \mathcal{X} \in \{A,B,C,D\}
			\label{cons:DC_form} \\
			&& \phantom{\mathrm{s.t.}}\quad
			\mathbf{v}_1^L \preceq \mathbf{A}_4^U \, \mathbf{v}_1 \preceq \mathbf{v}_1^U
			\label{cons:tau} \\
			&& \phantom{\mathrm{s.t.}}\quad
			\tfrac{1}{4} \mathbf{A}_4^U \, \mathbf{v}_1 = \tfrac{1}{4} \bigl(\mathbf{A}_4^L\bigr)^{\top} \, \mathbf{v}_2.
			\label{cons:equality}
		\end{IEEEeqnarray}
	\end{subequations}
	\endgroup
	We can now approximate problem~\eqref{eq:optimization_problem_example} by applying CCP to problem~\eqref{eq:optimization_problem_transformed}.
	\section{Numerical Results}
	We generate $200$ independent instances with random geometries, task profiles, computing capacities, and channel realizations. We set $c=3{\times}10^{8}$~m/s, $k_B=1.38{\times}10^{-23}$~J/K, $R_E=6371$~km, and the noise power to $-100$~dBm.
	For the $G{\to}U$ link, $B_{GU}=10$~MHz, $P_G=0.1$~W, $(\alpha_L,\alpha_N)=(2,2.5)$, and $\xi=10$~dB \cite{huang2024joint,mckay2005capacity}. For the $U{\to}S$ link, $f_c=30$~GHz, $B_{US}=10$~MHz, $P_U=1$~W, $\delta=25$~dB, and $h_{\mathrm{LEO}}=800$~km.
	For the ISL, $f_c=30$~GHz, $B_{\mathrm{ISL}}=1$~GHz, $T_{\mathrm{sys}}=354.81$~K, and $P_S=1$~W \cite{huang2024joint}. We set $b=1$ and uniformly draw $H_{\mathrm{UAV}}\in[50,120]$~m, $r\in[50,500]$~m, $\theta\in[20^\circ,90^\circ]$, $d_{\mathrm{ISL}}\in[100,1500]$~km, $F\in[1,10]$~Gcycles, $G_t=G_r\in[10,40]$~dBi, and $L_p \in [0,3]$~dB~\cite[\S~5.12]{maral2020satellite}.
	The values $W\in[0.1,5]$~MB, $f_{G}\in[0.05,0.3]$~GHz, $f_{U}\in[0.3,1]$~GHz, and $f_{S_1},f_{S_2}\in[0.5,3]$~GHz are independently drawn within their respective intervals.
	
	In Figs.~\ref{fig:convergence}--\ref{fig:sweeppgw}, both SGP (sequential geometric programming~\cite{van2018joint}) and CCP (the convex--concave procedure~\cite{lipp2016variations}) are solved using MOSEK. 
    For comparison, Global denotes the global optimum of problem~\eqref{eq:optimization_problem_example} obtained from the BARON solver~\cite{sahinidis1996baron}. Both SGP and CCP terminate when the relative change in the objective value between consecutive iterations falls below $10^{-4}$. The error bars in Fig.~\ref{fig:ccpvssgpbaron} indicate the standard deviation across all instances.
	\begin{figure}[t]
		\centering
		\includegraphics[width=\columnwidth]{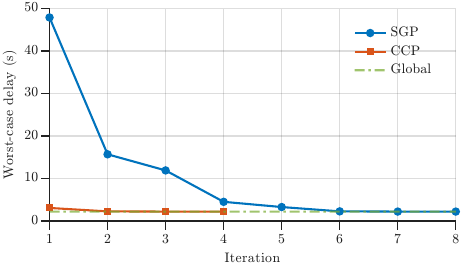}
		\caption{Convergence behavior.}
		\label{fig:convergence}
	\end{figure}
	
	Fig.~\ref{fig:convergence} shows the convergence of SGP and CCP with the BARON global optimum as a reference.
	Both SGP and CCP reach the worst-case delay of $2.218$~s, but CCP converges in $4$
	iterations versus $8$ for SGP. CCP also makes faster initial progress: it
	reaches $3.07$~s after one iteration, whereas SGP needs $5$ iterations to fall below $3.3$~s.
	\begin{figure}[t]
		\centering
		\includegraphics[width=\columnwidth]{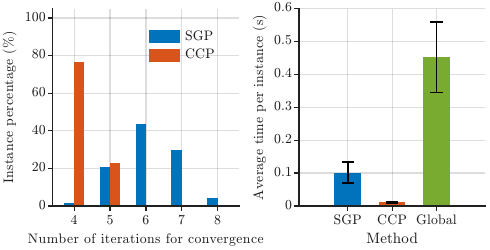}
		\caption{Iteration distribution and average solution time.}
		\label{fig:ccpvssgpbaron}
	\end{figure}
	
	As shown in Fig.~\ref{fig:ccpvssgpbaron}, CCP converges in fewer iterations than SGP, requiring an average of $4.24$ iterations compared with $6.14$ for SGP. CCP converges in $4$ iterations for $76.5\%$ of the instances, whereas SGP most frequently requires $6$ or $7$ iterations. CCP also reduces the average solution time from $0.1012$~s to $0.0113$~s, corresponding to an $8.9$-fold speedup over SGP. While both SGP and CCP approach the global optimum as shown in Fig.~\ref{fig:sweeppgw}, CCP with SOCP subproblems converges significantly faster, consistent with observations in~\cite{chen2023mixed}.
	\begin{figure}[t]
		\centering
		\includegraphics[width=\columnwidth]{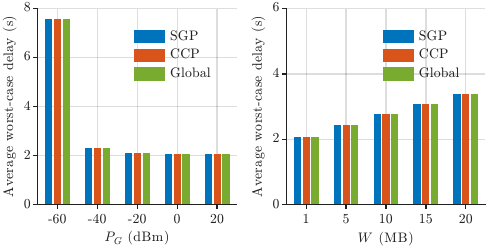}
		\caption{Average worst-case delay versus $P_G$ and $W$ for SGP, CCP, and BARON.}
		\label{fig:sweeppgw}
	\end{figure}
	
	In Fig.~\ref{fig:sweeppgw}, SGP and CCP achieve average worst-case delays nearly identical to the global solution across all tested settings. Increasing $P_G$ from $-60$ to $-40$~dBm reduces the average worst-case delay from about $7.6$~s to $2.3$~s. Increasing $P_G$ to $20$~dBm yields only a modest reduction to approximately $2.1$~s, indicating a diminishing improvement in delay at higher transmit powers. In contrast, increasing $W$ from $1$ to $20$~MB increases the average worst-case delay from about $2.1$~s to $3.4$~s. 
	\section{Conclusion}
This work demonstrates that QE facilitates the explicit construction of exact DC representations.
The resulting formulation uses SOCP subproblems and therefore removes the exponential-cone issue of SGP.
Numerical comparisons with BARON show that both CCP and SGP attain near-global performance, while CCP offers an $8.9$-fold speedup in average solution time.
Although evaluated on a four-node maritime route, Theorem~\ref{theorem_1} extends the exact construction to general $N$-factor sequential coupling. Embedded code generation and network optimization with multiple users, chains, and shared resources remain future work.

	\appendices
	\section{Proof of Theorem~\ref{theorem_1}} \label{Exact_proof}
	\begin{proof}
		Let $\varepsilon_j = (-1)^{a_j}\in\{\pm1\}$ for $j=2,\dots,N$, so that $\Omega=\prod_{j=2}^{N}\varepsilon_j$ and $S_N=\sum_{\boldsymbol{\varepsilon}\in\{\pm1\}^{N-1}} \bigl(\prod_{j=2}^{N}\varepsilon_j\bigr) \bigl(x_1+\sum_{j=2}^{N}\varepsilon_j x_j\bigr)^{N}$. We then fix $j\ge2$ and pair each sign pattern $\boldsymbol{\varepsilon}$ with the pattern obtained by flipping $\varepsilon_j$. The map  from $x_j$ to $-x_j$ exchanges the two paired summands while negating $\prod_{j\ge2}\varepsilon_j$, hence $S_N(x_1,\dots,-x_j,\dots,x_N)=-S_N(x_1,\dots,x_N)$. In other words, $S_N$ is odd in each of $x_2,\dots,x_N$. Consequently, every $x_1^{k_1}\cdots x_N^{k_N}$ of $S_N$ has $k_j$ odd for all $j\ge2$. Since $S_N$ is homogeneous of degree $N$, the $N{-}1$ odd exponents satisfy $\sum_{j=2}^{N}k_j=N-k_1\ge N{-}1$, forcing $k_1\le1$. Moreover, their sum has the parity of $N{-}1$, forcing $k_1=1$ and therefore $k_2=\dots=k_N=1$. Thus $S_N=C\prod_{n=1}^{N}x_n$ for some constant $C$. Using the multinomial expansion, the coefficient of $x_1x_2\cdots x_N$ in $\bigl(x_1+\sum_{j=2}^{N}\varepsilon_j x_j\bigr)^{N}$ is $N!\,\varepsilon_2\cdots\varepsilon_N$. Multiplying by  $\prod_{j=2}^{N}\varepsilon_j$ and using $\varepsilon_j^2=1$, each of the $2^{N-1}$ sign patterns contributes exactly $N!$, so $C=2^{N-1}N!$. Hence $S_N=2^{N-1}N!\prod_{n=1}^{N}x_n$, which is the first equality in \eqref{exact_solution}, and the second follows by negation.
	\end{proof}
	\section{Proof of Theorem~\ref{theorem_2}}
	\label{matrix_proof}
	\begin{proof}
		We index the rows of $\mathbf{A_N}$ by the bit strings
		$\mathbf{a}=(a_2,\dots,a_N)\in\{0,1\}^{N-1}$ and adopt the convention
		$a_1 \equiv 0$, so that the entry in row $\mathbf{a}$ and column $j$ is
		$(-1)^{a_j}$ for every $j=1,\dots,N$. This uniform description places the
		first column on the same footing as the rest, and the proof then reduces
		to evaluating inner products between columns.
		First, consider two distinct columns $j$ and $k$, where we assume
		$k \ge 2$ without loss of generality. Their inner product is
		$\sum_{\mathbf{a}\in\{0,1\}^{N-1}} (-1)^{a_j+a_k}$, which we evaluate by
		pairing each row with its counterpart obtained by toggling the single bit
		$a_k$. Toggling $a_k$ leaves $(-1)^{a_j}$ unchanged but reverses the sign
		of $(-1)^{a_k}$, so the two rows in each pair cancel and the columns are
		mutually orthogonal.
		Next, consider the case in which a column is paired with itself: every entry is $\pm 1$, so each of the $2^{N-1}$ terms equals $1$ and every
		column has squared norm $2^{N-1}$. Collecting the off-diagonal zeros from
		orthogonality and the diagonal value from the squared norm gives
		$\mathbf{A_N^\top}\mathbf{A_N} = 2^{N-1}\mathbf{I_N}$.
		Since $2^{N-1}\mathbf{I_N}$
		is nonsingular, the columns are linearly independent and
		$\mathrm{rank}(\mathbf{A_N})=N$. Since every eigenvalue of $\mathbf{A_N^\top}\mathbf{A_N}$ equals $2^{N-1}$, and the singular values of $\mathbf{A_N}$ are the square roots of these eigenvalues, every singular value $\sigma$ of $\mathbf{A_N}$ equals $\sqrt{2^{N-1}}$. The condition number reduces to $\kappa(\mathbf{A_N})=\sigma_{\max}/\sigma_{\min}=1$.
	\end{proof}
	\bibliographystyle{IEEEtran}
	\bibliography{bibliography}

@article{jung2023satellite,
  title={Satellite Clustering for Non-Terrestrial Networks: Orbital Configuration-Dependent Outage Analysis},
  author={Jung, Dong-Hyun and Ryu, Joon-Gyu and Choi, Junil},
  journal={IEEE Wirel. Commun. Lett.},
  volume={13},
  number={2},
  pages={550--554},
  month = {Feb.},
  year={2023}
}

@article{le2025performance,
  title={On Performance of Cooperative Satellite-{UAV}-secured Reconfigurable Intelligent Surface Systems with Phase Errors},
  author={Le, Anh-Tu and Vu, Thai-Hoc and Nguyen, Tan N and Minh, Bui Vu and Voznak, Miroslav},
  journal={IEEE Commun. Lett.},
    volume={29},
	number={4},
	pages={799-803},
	month = {Apr.},
  year={2025}
}

@article{shi2022delay,
  title={Delay minimization for {NOMA}-mm{W} scheme-based {MEC} offloading},
  author={Shi, Jia and Zhou, Yifan and Li, Zan and Zhao, Zhongling and Chu, Zheng and Xiao, Pei},
  journal={IEEE Internet Things J.},
  volume={10},
  number={3},
  pages={2285--2296},
  month = {Feb.},
  year={2023}
}

@book{horadam2012hadamard,
  title={Hadamard Matrices and Their Applications},
  author={Horadam, Kathy J},
  year={2012},
  publisher={Princeton University Press}
}

@article{park2017general,
  title={General heuristics for nonconvex quadratically constrained quadratic programming},
  author={Park, Jaehyun and Boyd, Stephen},
  journal={arXiv preprint arXiv:1703.07870},
  year={2017}
}

@article{davenport1988real,
  title={Real quantifier elimination is doubly exponential},
  author={Davenport, James H and Heintz, Joos},
  journal={J. Symbolic Comput.},
  volume={5},
  number={1-2},
  pages={29--35},
  month = {Feb.},
  year={1988}
}

@article{chen2023mixed,
  title={Mixed max-and-min fractional programming for wireless networks},
  author={Chen, Yannan and Zhao, Licheng and Shen, Kaiming},
  journal={IEEE Trans. Signal Process.},
  volume={72},
  pages={337--351},
  year={2024},
  publisher={IEEE}
}

@article{huang2024joint,
  title={Joint offloading and resource allocation for hybrid cloud and edge computing in {SAGIN}s: a decision assisted hybrid action space deep reinforcement learning approach},
  author={Huang, Chong and Chen, Gaojie and Xiao, Pei and Xiao, Yue and Han, Zhu and Chambers, Jonathon A},
  journal={IEEE J. Sel. Areas Commun.},
  volume={42},
  number={5},
   pages={1029--1043},
   month = {May},
  year={2024}
}

@book{boyd2004convex,
  title={Convex Optimization},
  author={Boyd, Stephen P and Vandenberghe, Lieven},
  year={2004},
  publisher={Cambridge University Press}
}

@article{lipp2016variations,
  title={Variations and extension of the convex--concave procedure},
  author={Lipp, Thomas and Boyd, Stephen},
  journal={Optim. Eng.},
  volume={17},
  number={2},
  pages={263--287},
  month = {Jun.},
  year={2016}
}

@article{sahinidis1996baron,
  title={{BARON}: a general purpose global optimization software package},
  author={Sahinidis, Nikolaos V},
  journal={J. Glob. Optim.},
  volume={8},
  pages={201--205},
  month = {Mar.},
  year={1996}
}

@article{filabadi2024exponential,
  title={Exponential Conic Relaxations for Signomial Geometric Programming},
  author={Filabadi, Milad Dehghani and Chen, Chen},
  journal={arXiv preprint arXiv:2406.05638},
  year={2024}
}

@article{ahmadi2018dc,
  title={{DC} decomposition of nonconvex polynomials with algebraic techniques},
  author={Ahmadi, Amir Ali and Hall, Georgina},
  journal={Math. Program.},
  volume={169},
  pages={69--94},
  month = {May},
  year={2018}
}

@article{bomze2004undominated,
  title={Undominated {DC} decompositions of quadratic functions and applications to branch-and-bound approaches},
  author={Bomze, Immanuel M and Locatelli, Marco},
  journal={Comput. Optim. Appl.},
  volume={28},
  pages={227--245},
  month = {Jul.},
  year={2004}
}

@inproceedings{murti2024lp,
  title={{LP}-based Construction of {DC} Decompositions for Efficient Inference of {M}arkov Random Fields},
  author={Murti, Chaitanya and Kashyap, Dhruva and Bhattacharyya, Chiranjib},
  booktitle={Proc. 27th Int. Conf. Artificial Intell. and Stat (AISTATS)},
  address={Valencia, Spain},
  pages={3781--3789},
  month = {May},
  year={2024},
  publisher={PMLR}
}

@inproceedings{mckay2005capacity,
  title={Capacity bounds for correlated {Rician} {MIMO} channels},
  author={McKay, Matthew R and Collings, Iain B},
  booktitle={Proc. ICC 2005 },
  volume={2},
  pages={772--776},
  month = {May},
  year={2005},
}

@article{maaz2025new,
  title={A new method for reducing algebraic programs to polynomial programs},
  author={Maaz, Muhammad and others},
  journal={arXiv preprint arXiv:2502.08210},
  year={2025}
}

@BOOK{maral2020satellite,
  author    = {Maral, G{\'e}rard and Bousquet, Michel and Sun, Zhili},
  title     = {Satellite Communications Systems: Systems, Techniques and Technology},
  publisher = {John Wiley \& Sons},
  year      = {2020}
}

@article{schaller2022embedded,
  title={Embedded code generation with {CVXPY}},
  author={Schaller, Maximilian and Banjac, Goran and Diamond, Steven and Agrawal, Akshay and Stellato, Bartolomeo and Boyd, Stephen},
  journal={IEEE Control Syst. Lett.},
  volume={6},
  pages={2653--2658},
  year={2022},
month = {May}
}

@ARTICLE{van2018joint,
  author  = {Van Chien, Trinh and Bj{\"o}rnson, Emil and Larsson, Erik G.},
  title   = {Joint Pilot Design and Uplink Power Allocation in Multi-Cell Massive {MIMO} Systems},
  journal = {IEEE Trans. Wireless Commun.},
  volume  = {17},
  number  = {3},
  pages   = {2000--2015},
  month   = {Mar.},
  year    = {2018},
  doi     = {10.1109/TWC.2017.2787702}
}
\end{document}